\pdfoutput=1
\RequirePackage{ifpdf}
\ifpdf 
\documentclass[pdftex]{sigma}
\else
\documentclass{sigma}
\fi

\numberwithin{equation}{section}

\newtheorem{Theorem}{Theorem}[section]
\newtheorem*{Theorem*}{Theorem}
\newtheorem{Corollary}[Theorem]{Corollary}
\newtheorem{Lemma}[Theorem]{Lemma}

\newtheorem{Conjecture}[Theorem]{Conjecture}
 { \theoremstyle{definition}

\newtheorem{Remark}[Theorem]{Remark} }

\usepackage{bbm}

\usepackage{xcolor}

\newcommand{\onframe}{\purple{a}}
\newcommand{\twodimN}{\purple{\mathbf{M}^2}}
\newcommand{\mygb}{\red{b}}

\newcommand{\ominf}{\blue{\omega}}

\newcommand{\mcone}{m_{c_1}}
\newcommand{\mctwo}{m_{c_2}}
\newcommand{\mca}{m_{c_a}}

\newcommand{\hhi}{{\red{\hat h_{i}}}}
\newcommand{\hRi}{{\red{\hat R_{i}}}}
\newcommand{\zui}{{\red{\hat u_{i}}}}

\newcommand{\zMone}{\red{{\hat M}_{1,i}}}

\newcommand{\zMtwo}{\red{{\hat M}_{2,i}}}

\newcommand{\hato}{\red{\omega}}

\newcommand{\hmv}{\red{\hat \omega}}

\newcommand{\repsilon}{ {\varepsilon}}

\newcommand{\rhozero}{\red{\rho_0}}

\newcommand{\purple}[1]{{\color{purple}#1}}

\newcommand{\horo}{\red{\mathfrak{h}}}

\newcommand{\hk}{\red{h_{k}}}

\newcommand{\ptcheck}[1]{\ptc{checked on #1}}

\newcommand{\blue}[1]{{\color{blue}#1}}
\newcommand{\red}[1]{{\color{red}#1}}

\newcounter{mnotecount}[section]

\renewcommand{\themnotecount}{\thesection.\arabic{mnotecount}}

\newcommand{\mnote}[1]
{\protect{\stepcounter{mnotecount}}$^{\mbox{\footnotesize
$
\bullet$\themnotecount}}$ \marginpar{
\raggedright\tiny\em
$\!\!\!\!\!\!\,\bullet$\themnotecount: #1} }

\newcommand{\jlcax}[1]{}
\newcommand{\eean}{\nonumber\end{eqnarray}}

\newcommand{\kk}[1]{}

\newcommand{\beq}{\begin{equation}}

\newcommand{\T}{{\mathbb T}}

\newcommand{\FS} 
 {F}

\newcommand{\HS} 
 {H_{\mbox{\scriptsize volume}}}

{\ptc{this should be removed in the oberwolfach version}}%

\newcommand{\zomega}{\mathring{\omega}}%
\newcommand{\eeal}[1]{\label{#1}\end{eqnarray}}
\newcommand{\bed}{\begin{deqarr}}
\newcommand{\eed}{\end{deqarr}}
\newcommand{\bedl}[1]{\begin{deqarr}\label{#1}}
\newcommand{\eedl}[2]{\arrlabel{#1}\label{#2}\end{deqarr}}

\newcommand{\mcU}{{\mycal U}}

\newcommand{\bel}[1]{\begin{equation}\label{#1}}
\newcommand{\bea}{\begin{eqnarray}}
\newcommand{\bean}{\begin{eqnarray}\nonumber}
\newcommand{\beal}[1]{\begin{eqnarray}\label{#1}}
\newcommand{\eea}{\end{eqnarray}}

\newcommand{\be}{\begin{equation}}
\newcommand{\eeq}{\end{equation}}
\newcommand{\ee}{\end{equation}}
\newcommand{\beqa}{\begin{eqnarray}}
\newcommand{\eeqa}{\end{eqnarray}}
\newcommand{\beqan}{\begin{eqnarray*}}
\newcommand{\eeqan}{\end{eqnarray*}}
\newcommand{\ba}{\begin{array}}
\newcommand{\ea}{\end{array}}

\newcommand{\warn}[1]
{\protect{\stepcounter{mnotecount}}$^{\mbox{\footnotesize
$
\bullet$\themnotecount}}$ \marginpar{
\raggedright\tiny\em
$\!\!\!\!\!\!\,\bullet$\themnotecount: {\bf Warning:} #1} }

\newcommand{\R}{\mathbb R}

\newcommand{\N}{\mathbb N}

\newcommand{\ptc}[1]{\mnote{{\bf ptc:}#1}}

\newcommand{\beqar}{\begin{deqarr}}
\newcommand{\eeqar}{\end{deqarr}}

\newcommand{\beaa}{\begin{eqnarray*}}
\newcommand{\eeaa}{\end{eqnarray*}}

\newcommand{\hrho}{\hat\rho}

\newcommand{\mv}{\omega}
\newcommand{\mage}[1]{#1}

\DeclareFontFamily{OT1}{rsfs}{}
\DeclareFontShape{OT1}{rsfs}{CGNPm}{n}{ <-7> rsfs5 <7-10> rsfs7 <10-> rsfs10}{}
\DeclareMathAlphabet{\mycal}{OT1}{rsfs}{CGNPm}{n}

\renewcommand{\ptcheck}[1]{}
\renewcommand{\red}[1]{#1}

\renewcommand{\blue}[1]{#1}
\renewcommand{\purple}[1]{#1}

\begin{document}
\allowdisplaybreaks

\renewcommand{\thefootnote}{}

\newcommand{\arXivNumber}{2207.14563}

\renewcommand{\PaperNumber}{005}

\FirstPageHeading

\ShortArticleName{On Asymptotically Locally Hyperbolic Metrics with Negative Mass}

\ArticleName{On Asymptotically Locally Hyperbolic Metrics\\ with Negative Mass\footnote{This paper is a~contribution to the Special Issue on Differential Geometry Inspired by Mathematical Physics in honor of Jean--Pierre Bourguignon for his 75th birthday. The~full collection is available at \href{https://www.emis.de/journals/SIGMA/Bourguignon.html}{https://www.emis.de/journals/SIGMA/Bourguignon.html}}}

\Author{Piotr T.~CHRU\'SCIEL~$^{\rm a}$ and Erwann DELAY~$^{\rm b}$}

\AuthorNameForHeading{P.T.~Chru\'{s}ciel and E.~Delay}

\Address{$^{\rm a)}$~Faculty of Physics, University of Vienna, Boltzmanngasse 5, A~1090 Vienna, Austria}
\EmailD{\href{mailto:piotr.chrusciel@univie.ac.at}{piotr.chrusciel@univie.ac.at}}
\URLaddressD{\url{https://homepage.univie.ac.at/piotr.chrusciel}}

\Address{$^{\rm b)}$~Laboratoire de Math\'ematiques d'Avignon, Avignon Universit\'e, F-84916 Avignon\\
\hphantom{$^{\rm b)}$}~and F.R.U.M.A.M., CNRS, F-13331 Marseille, France}
\EmailD{\href{mailto:erwann.delay@univ-avignon.fr}{erwann.delay@univ-avignon.fr}}
\URLaddressD{\url{https://erwanndelay.wordpress.com}}

\ArticleDates{Received August 01, 2022, in final form January 17, 2023; Published online January 23, 2023}

\Abstract{We construct families of asymptotically locally hyperbolic Riemannian metrics with constant scalar curvature (i.e., time symmetric vacuum general relativistic initial data sets with negative cosmological constant), with prescribed topology of apparent horizons and of the conformal boundary at infinity, and with controlled mass. In particular we obtain new classes of solutions with negative mass.}


\Classification{53C21; 83C05}

\begin{flushright}
\begin{minipage}{65mm}
\it Dedicated to Jean--Pierre Bourguignon\\ on the occasion of his 75th birthday
\end{minipage}
\end{flushright}

\renewcommand{\thefootnote}{\arabic{footnote}}
\setcounter{footnote}{0}

\section{Introduction}
Jean--Pierre Bourguignon made lasting contributions to differential geometry, to French mathematics, and to European research. Einstein metrics and their deformations are part of his research interests. This work is concerned with deformations of initial data for Lorentzian Einstein metrics, and it is a pleasure to dedicate to him this contribution to the subject.

In recent work~\cite{CDW} we derived a formula for the mass of three-dimensional asymptotically locally hyperbolic (ALH) manifolds obtained by gluing together two such manifolds ``at infinity''.
(This procedure is also known as ``Maskit gluing'', with the name introduced in~\cite{MazzeoPacardMaskit}.) We used the formula to prove {the main result there, namely} existence of conformally compactified manifolds without boundary, or with a toroidal black hole boundary, with conformal infinity of genus larger than {or} equal to 2, with constant scalar curvature (CSC) and with negative total mass.

The first step in~\cite{CDW} was to use a glue-in of an exactly hyperbolic region near infinity
as done in~\cite{ChDelayExotic}, introducing a small perturbation parameter $\epsilon$, which will be referred to as the ``exotic-gluing parameter''. The key point of the analysis in~\cite{CDW} was to control the limit of the mass when~$\epsilon$ tends to zero. This was done for a symmetric gluing of
two ALH manifolds with identical toroidal boundaries at conformal infinity.

The object of this work is to extend the analysis of~\cite{CDW} to the gluing of any two CSC
ALH manifolds
with non-spherical topology at infinity. {Thus, we can control the mass for gluings that do not have to be mirror-symmetric anymore, the manifolds being glued do not have to be identical, they can contain arbitrarily many black holes with arbitrary topology, and they are allowed to have more complicated topology at infinity.}
We show that the mass of the manifold obtained by connecting-at-infinity two such manifolds tends to a well-defined limit when the exotic-gluing parameter $\epsilon$ tends to zero; {see \eqref{18XI22.1} below, which generalises the formula proved in~\cite{CDW} for two identical components.
This} formula allows one to control the sign of the mass, obtaining in particular the following result, where the manifold resulting from the gluing has at least two boundary components, one at infinity with genus $ {{\mathbf g}_\infty}$, and another one at finite distance with genus ${{\mathbf g}_{\mathrm{BH}}}$ (where ``BH'' stands for ``black hole''):

\begin{Theorem} \label{T7VI22.1}
 Let ${{\mathbf g}_{\mathrm{BH}}}, {{\mathbf g}_\infty} \in \N$, ${{\mathbf g}_\infty} \ge {{\mathbf g}_{\mathrm{BH}}}$, with ${{\mathbf g}_\infty} \ge 2$ if ${{\mathbf g}_{\mathrm{BH}}}=0$.
There exist conformally compactifiable ALH manifolds of constant scalar curvature with a boundary of genus ${{\mathbf g}_{\mathrm{BH}}}$ with vanishing mean curvature, {with a conformal boundary at infinity of genus} ${{\mathbf g}_\infty} $, and with mass of {any} prescribed sign.
\end{Theorem}

The condition of constant scalar curvature corresponds to vacuum general relativistic time-symmetric initial data sets; an identical construction can be done for initial data sets for time-symmetric data sets with prescribed energy density.

We note that boundaries with vanishing mean curvature typically lie inside, or at the boundary, of the intersection of a black hole region with a time-symmetric initial-data slice.

The restriction ${{\mathbf g}_\infty} \ge {{\mathbf g}_{\mathrm{BH}}}$ is necessary, cf.~\cite{GSWW}.

Theorem~\ref{T7VI22.1} is a slightly less precise version of Corollary~\ref{NoSC21V22.1a} below. This last corollary follows immediately from Theorem~\ref{NoST29VII21.1} below, which is the main result of this paper, and the proof of which occupies most of the remainder of this paper.

The question of controlling the mass when gluing-at-infinity two manifolds across a single neck, with one manifold having spherical topology at infinity and the other not, remains to be settled.

\section{Maskit gluing at general boundaries}
\label{NoSs29VII21.1}

We adapt and extend the arguments in \cite{CDW} to accomodate general conformal boundaries at infinity. A useful device used in \cite{CDW} was to glue together two identical copies of a single manifold in a mirror-symmetric way; the associated simplifications do not arise in our context, which creates various difficulties that we address here.

The notations of that last reference are used throughout. {The current work draws heavily on constructions in~\cite{CDW}, some of which are only mentioned or sketched here, but we give a detailed presentation of those steps of the analysis in~\cite{CDW} which require substantial modifications.}

In the case of two summands, the manifold $(M,g)$ will be obtained by a boundary-gluing of two three-dimensional ALH manifolds, $(M_1,g_1)$ and $(M_2,g_2)$. We assume existence of a~coordinate system near each conformal boundary at infinity in which the metric $g_a$ takes the form
\begin{gather}\label{NoS21V22.1}
 g_a = \underbrace{
 \frac{{\rm d}r^2}{r^2+\kappa_a} + r^2 h_{\kappa_a}
 }_{\mygb} + e_a ,
 \qquad
 r\ge r_0, \quad
 a=1,2,
\end{gather}
for some $r_0>0$, where
\begin{equation}\label{NoS21V22.3}
 |e_a|_{\mygb}
 + |D e_a|_{\mygb}
 + \big|D^2 e_a\big|_{\mygb}
 \le C r^{-\sigma}
\end{equation}
with constants $\sigma>5/2$ and $C>0$, where $h_\kappa$ has constant Gauss curvature $\kappa \in \{0,\pm1\}$
and where $D$ is the covariant derivative operator of $g$.
We use the subscript $\mygb$ on a norm to indicate that the norm is taken using the metric $\mygb$.
Equation~\eqref{NoS21V22.3} holds, with $\sigma = 3$, both for the Birmingham--Kottler metrics
\begin{equation}\label{21XI22.p1}
 g_{\text{BK}} = \frac{{\rm d}r^2}{r^2+\kappa_a - 2m_c/r} + r^2 h_{\kappa_a},
\end{equation}
and for the Horowitz--Myers metrics
\begin{equation}\label{21XI22.p2}
 g_{\text{HM}}
 = \frac{{\rm d}r^2}{r^2 - 2m_c/r} + \big(r^2 - 2m_c/r\big) {\rm d}\theta^2 + r^2 {\rm d}\psi^2,
\end{equation}
where $(\theta,\psi)$ are periodic coordinates on $S^1\times S^1$ (cf.~\cite{Birmingham,HorowitzMyers,Kottler} or \cite{ChruscielBHBook}). Here $m_c\in \R$ is a~parameter which we call the \emph{coordinate mass}.
The reader will have noticed that when referring to the Birmingham--Kottler metrics or the Horowitz--Myers metrics we mean the space-part of these metrics, i.e., the metric induced on the static slices of the associated Lorentzian metrics.

For metrics of the form \eqref{NoS21V22.1} the mass of a connected component of the conformal boundary at infinity, which we denote for simplicity in \eqref{18IX20.4} by $\partial M$, is defined
 by the formula
 \cite{HerzlichRicciMass} (compare \cite[equation~(IV.40)]{BCHKK})
 \begin{gather}
 m( \partial M)
 =
 - \lim_{r\rightarrow \infty}\int_{\{r\} \times \partial M} D^j V
 \left( R{}^i{}_j - \frac {R{}}{n}\delta^i_j\right) {\rm d}\sigma_i,
 \label{18IX20.4}
 \end{gather}
 with the function $V$ given by
 \begin{equation}\label{1XI22.11}
 V =\sqrt{r^2 + \kappa_a}
 \end{equation}
 in the coordinate system of \eqref{NoS21V22.1}--\eqref{NoS21V22.3}.
 Here
 ${\rm d}\sigma_i := \sqrt{\det g} \,{\rm d}S_i$,
 $R_{ij}$ is the Ricci tensor of the metric $g$, $R$ its trace, and
 we have ignored an overall dimension-dependent positive multiplicative factor which is typically included in the physics literature.

 The mass $m$ of the metrics \eqref{21XI22.p1} is proportional to $m_c$, and that of the metrics \eqref{21XI22.p1} is proportional to $-m_c$.

We make appeal to the construction described in \cite[Section~2]{CDW}, where the hyperbolic metric has been glued-in within an $\epsilon$-neighborhood of boundary points $p_a\in \partial M_a$, {without changing the original metric away from the gluing region}.
We use the coordinates of \eqref{NoS21V22.1}--\eqref{NoS21V22.3}
with isothermal
polar coordinates for the boundary metric
\begin{equation}\label{26V22.1}
 h_{\kappa_a} ={\rm e}^{\zomega_a}
 \big(\underbrace{{\rm d}\rho^2 + \rho^2 {\rm d}\varphi^2}_{=:h_0} \big)
\end{equation}
on $\partial M_a$, and with $p_a$ located at the origin of these coordinates, with the conformal factor chosen so that $h_{\kappa_a}$ has constant Gauss curvature $\kappa_a\in \{0,\pm 1\}$.
Such coordinates can always be defined, covering a disc $D(\rhozero )$ centered at $p$ for some $\rhozero >0$, with the same coordinate radius $\rhozero $ for both boundaries.

After the exotic gluing has been performed, the metric $g_1$ is the original metric outside the half-ball $\mcU_{1,2\epsilon}$ of coordinate radius $2\epsilon<\rhozero$,
and is exactly hyperbolic inside the half-ball $\mcU_{1, \epsilon}$ of coordinate radius $\epsilon$
(see Figure~\ref{F26VIII21.2}); similarly for $g_2$.
 \begin{figure}[t]	\centering
\includegraphics[scale=0.8]{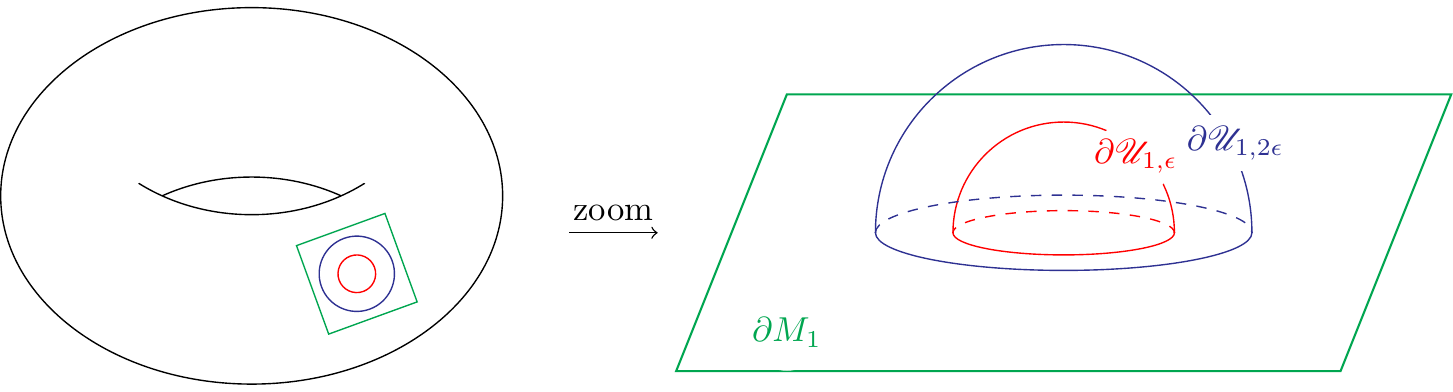}
 \caption{The sets $\mcU_{1, \epsilon}\subset \mcU_{1,2\epsilon}$ and their boundaries when the boundary at infinity $\partial M_1$ is a torus.
 {The parameter $\epsilon$ needs to be small to ensure convergence, in the construction of~\cite{ChDelayExotic}, of masses to the initial ones. The gluing of the conformal metrics at the boundary takes place within a disc of radius $1/i \le \epsilon/8$, with $i\to \infty$ as needed to ensure control of the error terms arising from a change of the conformal factor.}
 Figure from~\cite{CDW}.}	\label{F26VIII21.2}
\end{figure}
In order to control the mass we will need to consider a family of boundary gluings indexed by a parameter $\N\ni i\to\infty$.
For definiteness for $i\ge \red{8}/\epsilon$ we choose the hyperbolic hyperplanes $\blue{\horo_{a,i}}\subset M_a$ of \cite[Section~2]{CDW}
to be half-spheres of radius $1/i$ centered at the origin of the coordinates \eqref{26V22.1}. We choose any pair $(\Lambda_1,\Lambda_2)$ of isometries of the hyperbolic plane as in \cite[Section~2]{CDW} to obtain the boundary-glued manifold $M:=M_{\Lambda_1,\Lambda_2}$.

The above description generalises in an obvious way to gluings around any finite number of points at the conformal boundaries at infinity and any finite number of summands; the differences are purely notational.

Before stating our main theorem it is useful to recall the following:
Consider a two-dimen\-sio\-nal compact oriented manifold $\big(\twodimN ,h \big)$ and a finite number of distinct points {$p_k \in \twodimN $, $k=1,\dots, n$; } when $\twodimN $ is a sphere one needs $n\ge 3$.
There exists on $\twodimN \setminus \{p_k\}_{k=1}^n$ a smooth function $\ominf $, with \emph{puncture singularities, or cusps,} at the points $p_k$, such that the metric ${\rm e}^{\ominf }h $ is complete, has constant Gauss curvature equal to $-1$, and such that $\big(\twodimN ,{\rm e}^{\ominf }{h} \big)$ has finite total area; compare~\cite[Proposition~2.3]{GuillarmouSurfaces}, \cite{ZhangSphere}, and references therein.
An artist's impression of a~punctured torus can be seen in Figure~\ref{F21XI21.1}.

 \begin{figure}[t] \centering
 \includegraphics[scale=0.8]{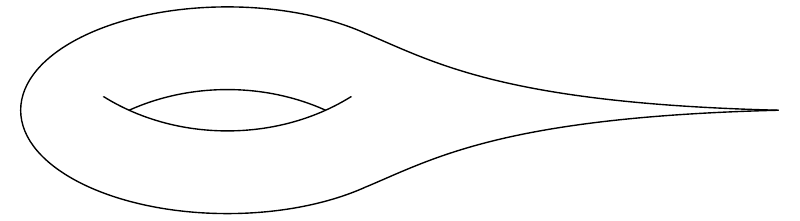}
 \caption{A punctured torus with a hyperbolic metric, from~\cite{CDW}. The figure fails to represent properly that the cusp region is infinitely long.}	\label{F21XI21.1}
\end{figure}

We claim:

\begin{Theorem}\label{NoST29VII21.1}
Let $N\ge 2$ and consider $N$ three-dimensional ALH manifolds $(M_a,g_a)$, $a=1,\dots, N$, with constant scalar curvature, and with a metric of the form \eqref{NoS21V22.1}--\eqref{NoS21V22.3}.
Let $p_{a,k}\in \partial M_a$, with $k=1,\dots,n_a \in \N$, where we assume that $n_a\ge 3$ when $\kappa_a =1$, and that each point $p_{a,k}\in \partial M_a$ has a unique partner $p_{b,j} \in \partial M_b$ {distinct from $ p_{a,k}$}.
 Let ${\rm e}^{\ominf_a } h_{\kappa_a}$ be the unique metric with scalar curvature equal to $-2$ on $\partial M_a \setminus \{p_{a,1},\dots,p_{a,n_a}\}$ with a cusp at {each} $p_{a,k}$. The mass of the Maskit-glued metric as described above converges, as $\epsilon$ tends to zero and $i$ tends to infinity, to the finite limit
\begin{equation}\label{18XI22.1}
 -  \sum_{a=1}^{N}
 \lim_{r\rightarrow \infty}\int_{\{r\} \times \partial M_a}
 D^j \big({\rm e}^{-\ominf_a/2}r \big)
 \left( R{}^\ell{}_j - \frac {R{}}{3}\delta^\ell_j\right)
{\rm d}\sigma_\ell.
\end{equation}
In the $a$th summand $R_{\ell j}$ denotes the Ricci tensor of the metric $g_a$, and $R$ its scalar curvature.
\end{Theorem}

\begin{Remark} \label{R12VI22.1}
 Suppose that we have
 \begin{equation}\label{29VII21.3}
 e_{ij} = r^{-3} \mu_{ij} + o \big(r^{-3}\big)
 \end{equation}
 for each summand,
 where the $\mu_{ij}$'s depend only upon the coordinates $x^A$ on $\partial M$. {In a~$b$-or\-thonor\-mal frame $(\onframe_1,\onframe_2,\onframe_3)$ with $\onframe_3$ proportional to $\partial_r$, formula} \eqref{18XI22.1} simplifies to
 \begin{gather*}
 \sum_{a=1}^{N} \int_{ \partial M_a}{\rm e}^{-\ominf_a/2}
 \left(
 2 \mu_{33}
 + 3 \sum_{i=1}^{2}\mu_{ii}
 \right) {\rm d}^{2}\mu_{\hk}.
 \end{gather*}
\end{Remark}

{Before passing to the proof of Theorem~\ref{NoST29VII21.1}, we note that the theorem
 implies existence of constant-scalar-curvature asymptotically-hyperbolic metrics with arbitrary total mass, and with prescribed topology both of black-hole boundaries and of conformal infinity:}

\begin{Corollary}\label{NoSC21V22.1a}
There exist three-dimensional conformally compactifiable ALH manifolds with constant scalar curvature, mass of any {prescribed value in $\R$}, and
\begin{enumerate}\itemsep=0pt
 \item[$1)$] a boundary at finite distance of genus {${{\mathbf g}_{\mathrm{BH}}}\ge 1$}
 with zero mean curvature and a conformal boundary at infinity of any genus ${{\mathbf g}_\infty}$ larger than ${{\mathbf g}_{\mathrm{BH}}}$;
 \item[$2)$] a spherical boundary at finite distance with zero mean curvature and a conformal boundary at infinity of any genus larger than or equal to two.
\end{enumerate}
\end{Corollary}

\begin{proof}
We start by noting that the contribution to the mass of a Bir\-ming\-ham--Kottler component, say $M_1$, with mass parameter which we denote by $\mcone $, can be written in the following simpler form in the limit when the gluing parameter $\varepsilon$ goes to zero and $i$ goes to infinity:
\[
 2 \mcone \int_{\partial M_1} {\rm e}^{-\ominf_1 /2}\, {\rm d}\mu_{h_{\kappa_1}}.
\]
When a component, say $M_2$, which is being glued is (the space-part of) a Horowitz--Myers metric with mass parameter {denoted as} $\mctwo$, its contribution to the mass in \eqref{18XI22.1}, again in the limit when the gluing parameter $\varepsilon$ goes to zero and $i$ goes to infinity, can be simplified to
\[
 - \mctwo \int_{\T^2} {\rm e}^{-\ominf_2 /2}\, {\rm d}\mu_{h_0}.
\]

1. Apply Theorem~\ref{NoST29VII21.1} to a Maskit gluing of a Birmingham--Kottler solution, with minimal boundary of genus {${{\mathbf g}_{\mathrm{BH}}}$} and mass parameter {$\mcone >m_{\text{crit}}$}, to ${{\mathbf g}_\infty}-{{\mathbf g}_{\mathrm{BH}}}$ Horowitz--Myers metrics with mass parameters $\mca >0$, where $a=2,\dots,{{\mathbf g}_\infty}-{{\mathbf g}_{\mathrm{BH}}}$.
{Here $m_{\text{crit}}=m_{\text{crit}}({{\mathbf g}_{\mathrm{BH}}})\le 0$ is the lower bound for the mass of a Birmingham--Kottler solution as needed for regularity.} The resulting limiting mass is
\[
 m= 2 \mcone \int_{\partial M_1} {\rm e}^{-\ominf_1 /2} \,{\rm d}\mu_{h_{\kappa_1}}
 - \sum_{a=2}^{{{\mathbf g}_\infty}-{{\mathbf g}_{\mathrm{BH}}}} \mca \int_{\T^2} {\rm e}^{-\ominf_a /2}\, {\rm d}\mu_{h_0},
\]
where the parameters $\mcone$ and $\mca$ are freely prescribable, {so that $m$ can take any values in $\R$}.

2. Let $(M_1,g_1)$ be obtained by a Maskit gluing, or an Isenberg--Lee--Stavrov~\cite{ILS} gluing, of a~spherical Birmingham--Kottler metric with a Horowitz--Myers metric. (In the Isenberg--Lee--Stavrov case the asymptotics \eqref{29VII21.3} is satisfied by~$g_1$.)
If it could be arranged that the resulting mass is negative, one would obtain a solution with toroidal conformal infinity and a spherical black hole; but the sign of the mass in this case is not clear. However, we can apply Theorem~\ref{NoST29VII21.1} to a Maskit gluing of $(M_1,g_1)$ to another Horowitz--Myers metric with sufficiently negative mass, which will provide the desired metric.
\end{proof}

\begin{Remark} \label{R9VI22}
One can use directly the construction of the proof of Theorem~\ref{NoST29VII21.1} to obtain a CSC ALH metric with a spherical boundary at finite distance with zero mean curvature (``apparent horizon''), negative mass, and a conformal boundary at infinity of any genus ${{\mathbf g}_\infty}$ larger than or equal to three,
by gluing a spherical Birmingham--Kottler metric across ${{\mathbf g}_\infty}$ punctures with ${{\mathbf g}_\infty}$ Horowitz--Myers metrics.
\end{Remark}

\begin{proof}[Proof of Theorem~\ref{NoST29VII21.1}]
We prove the result for the exotic Maskit gluing at one point
of each summand, {$p_1\in \partial M_1$ and $p_2 \in \partial M_2$},
in which case our assumptions require that neither summand is a sphere.
The proof in the {more} general case requires only {tedious} notational modifications.

Our aim is to prove
the existence of the limiting conformal factors $\omega_a$ on each summand $M_1$ and $M_2$. This was the contents of Lemma~5.9 in \cite{CDW}; the remaining arguments in \cite{CDW}, which do not need to be repeated here, establish~\eqref{18XI22.1}.

Some comments on the proof might be in order. The existence of the $\omega_a$'s is established by showing first a uniform upper bound on the sequence of conformal factors, by comparison with suitable barriers. One then needs a uniform lower bound: this is obtained by rewriting the equation in a form
to which a Harnack inequality applies. One further exploits the fact that the area does not
concentrate near the gluing necks; this follows from a good choice of the upper barriers. Convergence of a subsequence on compact subsets of the punctured manifolds follows then by elliptic estimates. The fact that the limit is the conformal factor for a punctured hyperbolic metric could most likely be established directly with some extra work, using the estimates derived here and in \cite{CDW} together with the results and techniques of Ruflin~\cite{Rupflin}. We avoid this supplementary work by appealing to the Deligne--Mumford compactness.

We now pass to the details of the above.

For notational simplicity ``boundary'' in the rest of the proof denotes the conformal boundary at infinity. Note that all our constructions are localised near that last boundary, so that the part of the boundary which corresponds to black hole horizons plays no role whatsoever in what follows.

By construction the boundary $\partial M$ of the new manifold is the gluing of
\[
 \zMone := \partial M_1\setminus \mcU_{1, \red{1/i}} \equiv \partial M_1\setminus \red{D(}1/i)
 \qquad
 \mbox{with}
 \quad
 \zMtwo := \partial M_2\setminus \mcU_{2, \red{1/i}}
 \equiv \partial M_2\setminus \red{D(}1/i)
\]
across their boundaries.
We will often view both
 $\zMone $ and $\zMtwo $ as subsets of $\partial M$.

To make clear the differentiable structure on $\partial M$ it is convenient we introduce a new coordinate on $\red{D(}\rhozero )$,
\begin{equation}\label{25XI21.1}
 \hrho = \frac{ \red{\log} \big( \rhozero ^{-1} \rho\big)}{ \red{\log} ( \rhozero i )}+1,
\end{equation}
so that $\rho\in [1/i,\rhozero ]$ corresponds to $ \hrho\in [0,1]$. The flat metric
\[
 h_0 = {\rm d}\rho^2 + \rho^2 {\rm d}\varphi^2,
 \qquad
 \rho\in (0,\rhozero ]
\]
becomes
\begin{equation*}
 h_0 =\rho^2
 \bigg(
 \frac{{\rm d}\rho^2}{\rho^2}+ {\rm d}\varphi^2
 \bigg)
 = \rho^2
 \big( \red{\log}^2( \rhozero i) {\rm d}\hrho^2 + {\rm d}\varphi^2 \big)
,
 \qquad
 \hrho\in (-\infty,1]
.
\end{equation*}
The differentiable structure near the connecting neck on
$\partial M \approx \partial M_1 \# \partial M_2$ is defined by letting $(\hrho,\varphi)$ range over $(-1,1)\times S^1$, with $(\hrho,\varphi) $
defined as above on
$D(\rhozero) \setminus \red{\overline{D\big(\rhozero^{-1}i^{-2}\big)}}\subset \partial M_1$
and with
\begin{equation}\label{2VI22.41}
 (\hrho,\varphi) \ \text{identified with the coordinates} \ (- \hrho,-\varphi)
\end{equation}
defined as above on $D(\rhozero) \setminus \red{\overline{D\big(\rhozero^{-1}i^{-2}\big)}}\subset \partial M_2$.
 The set covered by these coordinates will be referred to as \emph{the neck region}.
Thus
\[
 \hat h = h_0 \ \text{on} \ D(\rho_0).
\]

We define on $\partial M_1\setminus \{p_1\}$ a smooth metric $\hat h $ in the conformal class of $h_{\kappa_1}$ which equals to $\rho^{-2}{{\rm e}^{-\zomega_a}}h_{\kappa_1}$ on $D(\rhozero)\setminus \{0\}$.

\begin{Remark}
For coherence with \cite{ChDelayHPETv1, CDW} we indicated that we use the method of \cite[Section~2]{CDW} to extend the metric from the conformal boundary to the interior of the manifold. A more direct way in the current context, which differs from that of \cite[Section~2]{CDW} by an isometry of the metric in the hyperbolic region, proceeds as follows:
On $M_a$, in the coordinates centred at $p_a$ in the region where the metric is exactly the hyperbolic metric
\[
b=\frac{{\rm d}\vec y\,^2+{\rm d}x^2}{x^2},
\]
points in $M_1$ of coordinate $(\vec y,x)$ with $\frac{1}{i^2\epsilon}<\sqrt{|\vec y|^2+x^2}<\epsilon$ can be identified with points in $M_2$ of coordinates (in the same range)
\[
 \frac1{i^2\big(|\vec y|^2+x^2\big)}(S\vec y,x),
\]
where, in order to preserve orientation, $S$ is a mirror symmetry with respect to the horizontal axis in the $\vec y$-plane.
This guarantees that the hyperbolic metrics match across the totally geodesic hyperplane $\sqrt{|\vec y|^2+x^2}=1/i$; recall that $1/i< \epsilon < \rho_0/2$.
\end{Remark}

So far the $i$-dependent coordinates of \eqref{25XI21.1} introduce an explicit, but of course only apparent, $i$-dependence in $\hat h$:
\[
\hat h = \red{\log}^2( \rhozero i) {\rm d}\hrho^2 + {\rm d}\varphi^2 \quad \text{on} \ (-\infty,1]\times S^1.
\]
The metric $\hat h $ on $\partial M_2\setminus \{p_2\}$ is defined in an analogous way.

We denote by $\hhi $ the metric on $\partial M$ obtained from $\hat h$, as defined on
$\partial M_1 \setminus D\big(\rhozero^{-1}i^{-2}\big) $ and $\partial M_2\setminus D\big(\rhozero^{-1}i^{-2}\big) $ above, by using the identification \eqref{2VI22.41} in the neck region. It should be clear that $\hhi$ depends upon $i$ because the $\hhi$-diameter of the neck region equals $2 \log(\rho i)$, and hence grows with $i$.

The metric $\hhi$ is
conformal to $h_{\kappa_a}$ on $\partial M_a\setminus \red{D\big(\rhozero^{-1}i^{-2}\big)}$. It coincides with the cylindrical metric $ \red{\log}^2(\rhozero i) {\rm d}\hrho^2 + {\rm d}\varphi^2 $ in the neck region of both summands of the connected sum, hence is smooth on $\partial M$.

Consider the conformal class of metrics on $\partial M$ induced by $g$.
This conformal class depends upon $i$ but is independent of the exotic-gluing parameter $\epsilon$, except for the requirement that $\epsilon \ge \red{8}/i$.
(This is due to the fact that the parameter $\epsilon$ only plays a role in the initial insertion of an exactly hyperbolic region into $(M_a,g_a)$. The resulting metrics on $M_a$ depend upon $\epsilon$ in the interior, but the conformal class of the metric on $\partial M_a$ remains unchanged.
The condition $\epsilon \ge \red{8}/i$ is innocuous, as we are only concerned with the limit $i\to\infty$.)
In this class there exists a unique metric with constant scalar curvature equal to minus two.
It can be found by solving the two-dimensional Yamabe equation
\begin{equation}\label{NoS31VII21.31}
 \Delta_{\hhi }u_i = - R {\rm e}^{ u_i} + \hRi,
\end{equation}
with $R=-2$, and where $\hRi $ is the scalar curvature of the metric $\hhi $, so that the metric ${\rm e}^{ u_i}\hhi $ has scalar curvature $R$.
It is important in what follows that the function $u_i$ is independent of the parameter $\epsilon$ introduced when gluing-in the hyperbolic metric near the points $p_i$.

The need to use a constant-scalar-curvature representative of the conformal class at infinity arises from the definition of mass. Indeed, it is built-in into \eqref{NoS21V22.1}--\eqref{1XI22.11} that the metric $h_{\kappa_a}$ has constant scalar curvature.

The Gauss--Bonnet theorem gives{\samepage
\begin{equation*}
 \int_{\partial M } {\rm e}^{u_i} \,{\rm d}\mu_{\hhi }= \red{2\pi}
 \big(
 2-\chi(\partial M_1) - \chi(\partial M_2)
 \big),
\end{equation*}
where $ {\chi}\big(\twodimN\big)$ denotes the Euler characteristic of a two-dimensional manifold $\twodimN$.}

For $0<a<b$ let
\[
\Gamma(a,b) := D(b) \setminus \overline{D(a)},
\]
where $D(c)$ denotes an open disc of radius $c$ in $\R^2$.
By construction,
there exists a function $\zui $ defined
on
\[
 \red{\Gamma_i}:=
 \Gamma\big(\rhozero^{-1} i^{-2}, \rhozero\big)
\]
so that there we have
\begin{equation*}
 \hhi ={\rm e}^{ \zui } h_0.
\end{equation*}
Then the metric
\[
 {\rm e}^{ u_i+\zui } h_0,
\]
defined on $\red{\Gamma_i}$, has scalar curvature equal to minus two.
Since $h_0$ is flat, the function
\begin{equation*}
 \hmv_i:= u_i+\zui
\end{equation*}
satisfies on $\red{\Gamma_i}$ the equation
\begin{equation}\label{NoS31VII21.32}
 \Delta_{h_0 }\hmv_i = 2 {\rm e}^{\hmv_i}.
\end{equation}

When $\kappa_{1}<0$ there exist on $M_1\setminus \red{D\big(}\rhozero^{-1} i^{-2}\big)$ two metrics of negative scalar curvature conformal to each other, namely the metric ${\rm e}^{u_i} \hhi $ and the original metric $h_{\kappa_1}$. We write
 \begin{equation}\label{30V22.91}
 {\rm e}^{u_i} \hhi = {\rm e}^{\hato_i} h_{\kappa_1}.
 \end{equation}
 The functions $\hato_i$ are solutions of the equation
\begin{equation*}
 \Delta_{h_{\kappa_1}} \hato_i = 2 {\rm e}^{\hato_i} + 2 \kappa_1 = 2 {\rm e}^{\hato_i} - 2
.
\end{equation*}
The maximum principle shows that $\hato_i$ has neither a positive
 interior maximum nor a negative interior minimum on the compact manifold with boundary $ \partial M_1 \setminus D(a)$ for $a\in\big[\blue{\rhozero^{-1} i^{-2}},\rho_0\big]$.

When $\kappa_{1}=0$
we write again \eqref{30V22.91}, except that now we have
\begin{equation*}
 \Delta_{h_{\kappa_1}} \hato_i = 2 {\rm e}^{\hato_i}.
\end{equation*}
The maximum principle ensures then the property, that $\hato_i$ has no interior \mage{maximum} on the compact manifold with boundary $ \partial M_1 \setminus D(a)$ for $a\in\big[\blue{\rhozero^{-1} i^{-2}},\rho_0\big]$.

It also follows from \eqref{26V22.1} that on the annulus $D(\rhozero )
 \setminus \overline{D\big(\blue{\rhozero^{-1} i^{-2}}\big)}$
we can rewrite \eqref{30V22.91} as
 \begin{equation*}
 {\rm e}^{\hmv_i} h_0
 = {\rm e}^{\hato_i} h_{\kappa_1}
 = {\rm e}^{\hato_i + \mathring \omega_1} h_{0}.
 \end{equation*}

We will need the property (cf., e.g.,~\cite{MazzeoTaylor}) that solutions of the equation
\begin{equation*}
 \Delta \mv= 2 {\rm e}^{\mv} +\varpi,
\end{equation*}
where $\varpi$ is a function independent of $\mv$ (in the cases of interest here $\varpi\equiv 0$ or $\varpi\equiv-2$, compare~\eqref{NoS31VII21.31}),
satisfy a comparison principle: given a conditionally compact domain with boundary:
\begin{equation*}
 \hat \mv > \mv \quad \text{on} \ \partial\Omega
 \
 \Longrightarrow
 \
 \hat \mv > \mv \quad \text{on} \ \Omega.
\end{equation*}

The following metric, which has constant negative scalar curvature equal to $-2$, provides a~useful comparison function:
\begin{equation}
 \label{NoS9XI21.3ab}
 {\rm e}^{\mv_{*,i}} h_0 :=
\Bigg(
 \frac{\pi }{
 \log (\rho_0^2i^2) \sin\Big( \pi \frac{\log (\rho/\rho_0) }{\log (\rho_0^2i^2)}\Big) \rho}
\Bigg)^2 \big({\rm d}\rho^2+\rho^2 {\rm d}\varphi^2\big), \qquad
 \rho\in \left(\frac1{\rho_0i^2},\rho_0 \right),
\end{equation}
The conformal factor ${\rm e}^{\mv_{*,i}}$ tends to infinity at $\rho=\rhozero^{-1} i^{-2}$ and at $\rho=\rhozero $.
In the coordinates $(\hat \rho, \varphi)$ the metric \eqref{NoS9XI21.3ab} reads
\begin{equation}
 \label{NoS25XI21.6}
 {\rm e}^{\mv_{*,i}}h_0 =
\left(
 \frac{\pi }{
 2 \cos( \pi \hrho/2 )
}
 \right)^2 \left({\rm d}\hrho^2+ \frac{1}{\red{\log^2(\rho_0i)}} {\rm d}\varphi^2\right)
, \qquad
 \hrho\in (-1,1).
\end{equation}
Note that the circle $\hrho =0$ is a closed geodesic minimising length for the metric \eqref{NoS25XI21.6}, of length $\pi^2/\red{\log (\rho_0i)}$.

Since the function ${\rm e}^{\mv_{*,i}}$, defined in \eqref{NoS9XI21.3ab}, tends to infinity as the boundary of $\Gamma_i$ is approached, the comparison principle gives:

\begin{Lemma}\label{NoSL23X21.1m}
On $D(\rhozero ) \setminus \overline{D\big(\rhozero^{-1} i^{-2}\big)}$ it holds that
\begin{equation} \label{NoS9XI21.2}
{\rm e}^{\hmv_i}\leq {\rm e}^{\mv_{*,i}}.
\end{equation}
\end{Lemma}

\begin{Remark} \label{NoSR22V22.1}
On the circle $\rho=1/i$ the function
${\rm e}^{\mv_{*,i}} $ tends to infinity as $i\to \infty$, but the metric length $\ell_i$ of $S_{1/i}$ equals
\[
\ell_i=\frac1i\int_{\varphi\in[0,2\pi]}\big({\rm e}^{\hmv_{i}/2}\big)_{|\rho=1/i}\, {\rm d} \varphi\leq\frac1i\int_{\varphi\in[0,2\pi]}\big({\rm e}^{\mv_{*,i}/2}\big)_{|\rho=1/i}\,{\rm d}\varphi=\frac{\pi^2}{\log{(\rho_0i)}},
\]
so that $\ell_i$ approaches zero as $i$ tends to infinity.
\end{Remark}

\begin{Corollary} \label{NoSC21XI21} For any $\rho_1\in\big(\rhozero^{-1} i^{-2},\rhozero \big)$,
 there exists a constant $\hat c =\hat c(\rho_1)$ such that
 \[
 \mv_i \le \hat c
 \]
 on $\partial M_1 \setminus D(\rho_1)$, independently of $i$.
\end{Corollary}

\begin{proof}
At $\rho=\rho_1$ we have
\begin{equation*}
{\rm e}^{\hmv_i}\leq {\rm e}^{\mv_{*,i}}=
\Bigg(
 \frac{\pi }{
 \log \big(\rho_0^2i^2\big) \sin\Big( \pi \frac{\log (\rho/\rho_0) }{\log (\rho_0^2i^2)} \Big) \rho
}
 \Bigg)^2
 \longrightarrow_{i\rightarrow +\infty}\frac{1}{\rho_1^2\log^2(\rho_1/\rhozero)}.
\end{equation*}
This shows that
the $\mv_i$'s are bounded by a constant $\hat c(\rho_1 ) >0 $ independently of $i$ on $S(\rho_1)$ for $\rho_1\in\big(\rhozero^{-1} i^{-2},\rhozero \big)$.
The result follows now from the maximum principle.
\end{proof}

The corollary gives an estimation of the conformal factors from above. In order to prove convergence of the sequence $\mv_i$ away from the puncture, we also need to bound the sequence of conformal factors away from zero. As a tool towards this we
consider the sequence of ``half-areas'':
\begin{equation*}
 0 < A_{1,i}:=\int_{\zMone } {\rm e}^{\mv_i} \,{\rm d}\mu_{h_{\kappa_1}} =
 \int_{\zMone } {\rm e}^{u_i} \,{\rm d}\mu_{\hhi }
 <
 \int_{\partial M } {\rm e}^{u_i} \,{\rm d}\mu_{\hhi }
 = \red{2\pi}
 \big(
 2-\chi(\partial M_1) - \chi(\partial M_2)
 \big).
\end{equation*}
Thus the sequence $\{A_{1,i}\}_{i\in \N}$ is bounded, and so passing to a subsequence $\{i_j\}_{j\in \N}$, if necessary, we can assume that the limit exists:
\begin{equation*}
 A_{1}:= \lim_{j\to\infty} A_{1,i_j}.
\end{equation*}
Using analogous definitions on $\partial M_2$, we have
\begin{equation*}
 A_{1,i} + A_{2,i} = \red{2\pi}
 \big(
 2-\chi(\partial M_1) - \chi(\partial M_2)
 \big)
 \
 \Longrightarrow
 \
 A_2 = \red{2\pi}
 \big(
 2-\chi(\partial M_1) - \chi(\partial M_2)
 \big) - A_1.
\end{equation*}
It follows that at least one of $A_1$ and $A_2$ is not zero. Exchanging $M_1$ with $M_2$, we can without loss of generality assume that
\begin{equation}\label{4VI22.1}
 A_1 \ne 0.
\end{equation}

In our next result the parameter $\repsilon>0$ should not be confused with the parameter $\epsilon$ introduced by the exotic gluing of $M_1$ with $M_2$:

\begin{Lemma}\label{NoSL23X21.1}Assuming \eqref{4VI22.1}, there exist constants $C_1$ and $C_2$ such that, for all $\repsilon$ sufficiently small,
\begin{equation*}
 \limsup_{j\in \N}
 \Big(
 \inf_{\red{\partial M_1 \setminus D(\repsilon/2)}} \mv_{i_j}
 \Big)
 \le C_1,
 \qquad
 C_2\le
 \liminf_{j\in\N}
 \Big(
 \sup_{\red{\partial M_1 \setminus D(\repsilon/2)}} \mv_{i_j}
 \Big).
\end{equation*}
In other words, if $\repsilon$ is sufficiently small, then for all $j$ sufficiently large we have
\begin{equation}\label{NoS15VIII21.21b}
 \inf_{\red{\partial M_1 \setminus D(\repsilon/2)}} \mv_{i_j}
 \le C_3,
 \qquad
 C_4\le
 \sup_{\red{\partial M_1 \setminus D(\repsilon/2)}} \mv_{i_j},
\end{equation}
for some constants $C_3$ and $C_4$.
\end{Lemma}

\begin{proof}
It holds
\begin{equation*}
 \int_{\Gamma(a,b)} {\rm e}^{\mv_{*,i} }\,{\rm d}\mu_{h_0}
 = -\frac{2 \pi ^2 \cot \Big(\frac{\pi
 \log (\rho/\rho_0 )}{\log
(\rho_0^2i^2)}\Big)}{\log
(\rho_0^2i^2)}\Bigg|^b_a,
\end{equation*}
and note that
\begin{equation}\label{NoS18XI21.1b}
 \int_{\Gamma(1/i,\repsilon/2)} {\rm e}^{\mv_{*,i} }\, {\rm d}\mu_{h_0}
 = \frac{2 \pi ^2 \cot \Big(\frac{\pi
 \log (2/(\rho_0\repsilon))}{\log
 (\rho_0^2i^2)}\Big)}{\log
 \big(\rho_0^2i^2\big)}
 \to_{i\to\infty}
 \frac{2\pi }{\log \big(2/(\rho_0\repsilon)\big)}
 \to_{\repsilon\to 0} 0.
\end{equation}
We have, for all $\epsilon< 2\rhozero$ and for all $i$ larger than $2/\repsilon$,
\begin{align}
 A_{1,i}
 &=
 \int_{\zMone} {\rm e}^{\mv_i} \,{\rm d}\mu_{h_{\kappa_1}}= \int_{\red {\partial M_1 \setminus D(\repsilon/2)}} {\rm e}^{\mv_i}\, {\rm d}\mu_{h_{\kappa_1}}
 +
 \int_{\Gamma(1/i,\repsilon/2) } {\rm e}^{\mv_i}\, {\rm d}\mu_{h_{\kappa_1}}
 \nonumber
\\
 &=
 \int_{\red{\partial M_1 \setminus D(\repsilon/2)}} {\rm e}^{\mv_i}\, {\rm d}\mu_{h_{\kappa_1}}
 +
 \int_{\Gamma(1/i,\repsilon/2) } {\rm e}^{\hmv_i}\, {\rm d}\mu_{h_{0}}.\label{NoS6X21.2ax}
\end{align}
The estimate \eqref{NoS9XI21.2} shows that
\begin{gather*}
\int_{\Gamma(1/i,\repsilon/2) } {\rm e}^{\hmv_i} \,{\rm d}\mu_{h_0}
 \le
 \int_{\Gamma(1/i,\repsilon/2) } {\rm e}^{\mv_{*,i} }\, {\rm d}\mu_{h_0}. 
\end{gather*}
It follows from \eqref{NoS18XI21.1b} that there exists $\repsilon_0$ such that for all $2/{i} <\repsilon\le \repsilon_0$ the last term in \eqref{NoS6X21.2ax} is in $(0,A_{1,i}/2)$,
which implies
\begin{eqnarray*}
 \frac 12 A_{1,i} \le \int_{\red{\partial M_1 \setminus D(\repsilon/2)}} {\rm e}^{\mv_i} \,{\rm d}\mu_{h_{\kappa_1}}
 < A_{1,i}. 
\end{eqnarray*}
The conclusion readily follows from
\begin{gather*}
 \int_{\red{\partial M_1 \setminus D(\repsilon/2)}} {\rm e}^{\inf \mv_i }\, {\rm d}\mu_{h_{\kappa_1} }
 \le \int_{\red{\partial M_1 \setminus D(\repsilon/2)}} {\rm e}^{\mv_i} \,{\rm d}\mu_{h_{\kappa_1} }
 \le \int_{\red{\partial M_1 \setminus D(\repsilon/2)}} {\rm e}^{\sup \mv_i} \,{\rm d}\mu_{h_{\kappa_1}}.
 \tag*{\qed}
\end{gather*}
\renewcommand{\qed}{}
\end{proof}

We are ready now to prove the equivalent of Lemma~5.9 of \cite{CDW} on the summand chosen as in~\eqref{4VI22.1}.

\begin{Lemma} \label{NoSL23X21.2}
Under \eqref{4VI22.1}, there exists a smooth function
\[
 \mv_\infty\colon \ \partial M_1\setminus \{p_1\} \to \R
\]
such that a subsequence of $\{\mv_{i_j}\}_{j\in\N}$ converges uniformly to $\mv_\infty$ on every compact subset of $\partial M_1\setminus \{p_1\}$.
Similarly derivatives of any order of $\mv_{i_j}$ converge to derivatives of $\mv_\infty$, uniformly on every compact subset of $\partial M_1\setminus \{p_1\}$.
\end{Lemma}

\begin{proof}
The proof is an adaptation to our setting of the arguments given in~\cite{CDW}, we present here the details for completeness.

We will need a function $\psi\in C^\infty\big(M_1\setminus \{p_1\}\big)$ which satisfies the equation
\begin{equation}\label{1VI22.41}
 \Delta_{h_{\kappa_1}} \psi = -2 + c \delta_{p_1},
\end{equation}
where $\delta_{p_1}$ is the Dirac measure centered at $p_1$, with $c$ equal to twice the $h_{\kappa_1}$-area of $\partial M_1$. This choice of $c$ ensures existence of $\psi$, which can be seen as follows:
Let $(\rho,\varphi)$ be any coordinates near $p_1$ such that $h_{\kappa_1} = {\rm e}^{\omega_1}\big( {\rm d}\rho^2 + \rho^2 {\rm d}\varphi^2\big)$ there. Let $\psi_1\in C^\infty\big(\partial M_1\setminus \{p_1\}\big)$ be any function which equals $\ln \rho$ near $p_1$. There exists a constant $c_1\ne 0 $ such
that
\[
f:= \Delta_{h_{\kappa_1}} \psi_1 -c_1 \delta_{p_1}\in C^\infty(\partial M_1).
\]
Letting $\langle T,f\rangle$ denote the action of a distribution $T$ on a smooth function $f$, we have
\begin{equation*}
 0 = \big\langle \psi_1, \Delta_{h_{\kappa_1}} 1\big\rangle
 = \big\langle \Delta_{h_{\kappa_1}} \psi_1, 1\big\rangle
 = \int_{\partial M_1} f \, {\rm d}\mu_{h_{\kappa_1} } + c_1,
\end{equation*}
where $d\mu_{h_{\kappa_1} }$ is the measure associated with the metric ${h_{\kappa_1} }$.
Hence
\begin{equation*}
 \int_{\partial M_1} f \, {\rm d}\mu_{h_{\kappa_1} } = -c_1.
\end{equation*}
Consider the equation
\begin{equation*}
 \Delta_{h_{\kappa_1}} \psi_2 = -2 -\frac{c}{c_1} f.
\end{equation*}
By choice of $c$ the right-hand side has zero-average over $\partial M_1$, which guarantees existence of a~smooth function $\psi_2$ solving the equation. Then{\samepage
\begin{equation*}
 \psi := \frac{c}{c_1} \psi_1 + \psi_2
\end{equation*}
solves \eqref{1VI22.41}.}

To continue, as in \cite{CDW} we let $K$ be any compact subset of $\partial M_1\setminus \{p_1\}$. There exists $\rho_K>0$ such that $K\subset \partial M_1\setminus D(\rho_K)$. It thus suffices to prove the result with $ K=\partial M_1\setminus D(\rho_K)$, which will be assumed from now on.

Let $K_1 = \partial M_1\setminus D(\rho_K/2)$. Taking $\repsilon=\rho_K $ in Lemma~\ref{NoSL23X21.1} ensures that \eqref{NoS15VIII21.21b} holds on $K_1$ for all $i\ge i_1 $ for some $i_1<\infty$.

Let $i\ge i_1$. By Corollary~\ref{NoSC21XI21},
there exists a constant $c_1$, independent of $i$, such that
\begin{equation*}
\red{v_i}:=\psi -\mv_i \ge c_1 \quad \text{on} \ K_1.
\end{equation*}
On $K_1$ it also holds
\begin{equation}\label{NoS24X21.1b}
 c_2 \le \psi \le c_3,
\end{equation}
for some constants $c_2$ and $c_3$.
Define
\[
 \red{\hat v_i} : = \red{v_i} -c_1 +1.
\]
It holds that $\red{\hat v_i} \ge 1$ on $K_1$.

Moreover, $\red{\hat v_i}$ satisfies the equation
\begin{equation*}
 \Delta_{ h_{\kappa_1}} \red{\hat v_i} = \red{\psi_i}\red{\hat v_i},
\end{equation*}
where
\begin{equation*}
0 \ge \red{\psi_i} : =- 2 \frac{{\rm e}^{\mv_i}}{ \red{\hat v_i}}
 =- 2 \frac{{\rm e}^\psi {\rm e}^{\mv_i-\psi }}{ \red{\hat v_i}}
\ge - 2 {\rm e}^{c_3} {\rm e}^{\mv_i-\psi }
 \ge -2 {\rm e}^{c_3-c_1}.
\end{equation*}

By Harnack's inequality, there exists a constant $C_1= C_1(K,K_1) >0$ such that on $K$ we have
\begin{equation*}
 \sup_K \red{\hat v_i} \le C_1 \inf_{K_1} \red{\hat v_i}.
\end{equation*}
This, together with the definition of $\red{\hat v_i}$, shows that
\begin{equation}\label{NoS6VIII21.11a}
 \sup_K \red{v_i} \le C_1 \inf_{K_1} \red{v_i} +d_1 =
 C_1 \big({-}\sup_{K_1}( {\mv_i}-\psi) \big) +d_1,
\end{equation}
for some constant $d_1$.
Equation~\eqref{NoS15VIII21.21b} shows that there exists a constant $c_4$ such that
\begin{equation*}
- \sup_{K_1} (\mv_i -\psi)
 \le c_4.
\end{equation*}
From \eqref{NoS6VIII21.11a} we obtain
\begin{equation*}
 \sup_K (\red{v_{i }}- \psi) \le C_1 c_4 +d_1
 \
 \Longrightarrow
 \
 \inf_K \red{\mv_{i }} = - \sup_K (\red{v_{i }}- \psi) \ge -(C_1 c_4 +d_1).
\end{equation*}
This, together with \eqref{NoS24X21.1b}, shows that
that for every compact subset $K$ of $\partial M_1\setminus \{p_1\}$
there exists a constant $\hat C_K$ such that
\begin{equation*}
 - \hat C_K \le \red{\mv_{i }} \le \hat C_K.
\end{equation*}
Elliptic estimates, together with a standard diagonalisation argument, show that there exists a subsequence $\mv_{i_{j }}$ which
converges uniformly on every compact subset of $\partial M_1 \setminus \{p_1\}$ to a~solution~$\mv_\infty $
of \eqref{NoS31VII21.32} on $\partial M_1 \setminus \{p_1\}$. Convergence of derivatives follows again from elliptic estimates.
\end{proof}

To continue, we wish to show that $A_2\ne 0$. As a step towards this we claim that
\begin{equation}\label{4VI22.11rd}
 A_1 = \int_{\partial M_1} {\rm e}^{\mv_\infty}\, {\rm d}\mu_{h_{\kappa_1} }.
\end{equation}
In order to prove \eqref{4VI22.11rd}, for any $\repsilon>0$ we can write
\begin{align}
 \int_{\partial M_1} {\rm e}^{\mv_\infty}\, {\rm d}\mu_{h_{\kappa_1} }
 &=
 \int_{\partial M_1\setminus D(\repsilon/2)} {\rm e}^{\mv_\infty}\, {\rm d}\mu_{h_{\kappa_1} }
 +
 \int_{ D(\repsilon/2)} {\rm e}^{\mv_\infty} \,{\rm d}\mu_{h_{\kappa_1} }
 \nonumber
\\
 &=
 \int_{\partial M_1\setminus D(\repsilon/2)}\!\! ( {\rm e}^{\mv_\infty} - {\rm e}^{\mv_i} )\, {\rm d}\mu_{h_{\kappa_1} }\!
 +
 \underbrace{
 \int_{\partial M_1\setminus D(\repsilon/2)}\!\! {\rm e}^{\mv_i}\, {\rm d}\mu_{h_{\kappa_1} }\!
 +
 \int_{\Gamma(1/i,\repsilon/2) }\!\! {\rm e}^{\mv_i}\, {\rm d}\mu_{h_{\kappa_1}}
 }_{=A_{1,i}}
 \nonumber
\\
&\quad{}
 -
 \int_{\Gamma(1/i,\repsilon/2) } {\rm e}^{\mv_i}\, {\rm d}\mu_{h_{\kappa_1}}
 +
 \int_{ D(\repsilon/2)} {\rm e}^{\mv_\infty}\, {\rm d}\mu_{h_{\kappa_1} }. \label{8VI22.41}
\end{align}

Recall that
\begin{align}
 \int_{\Gamma(1/i,\repsilon/2) } {\rm e}^{\mv_i}\, {\rm d}\mu_{h_{\kappa_1}}
 & =
 \int_{\Gamma(1/i,\repsilon/2) } {\rm e}^{\hmv_i}\, {\rm d}\mu_{h_0}
\nonumber
\\
 &\le
 \int_{\Gamma(1/i,\repsilon/2) } {\rm e}^{\mv_{*,i} } \,{\rm d}\mu_{h_0}
 \to_{i\to\infty}
 \frac{2\pi }{\log (2/(\rho_0\repsilon))}
 \to_{\repsilon\to 0}
 0. \label{NoS6X21.2asf}
\end{align}

Let $\eta>0$. The last term in \eqref{8VI22.41} will be smaller than $\eta/4$ for all $\varepsilon$ small enough because ${\rm e}^{\omega_\infty}\in L^1(\partial M_1)$. Equation~\eqref{NoS6X21.2asf} shows that we can reduce $\varepsilon >0$ if necessary so that for all $i$ large enough the first term in the last line of \eqref{8VI22.41} will be smaller than $\eta/4$. Since $\omega_i $ converges to $\omega_\infty$ uniformly on the compact subset $\partial M_1\setminus D(\epsilon/2)$ of $\partial M_1\setminus \{p_1\}$, the first term in the second line of \eqref{8VI22.41} is smaller than $\eta/4$ for all $i$ large enough. For $j$ large enough it holds that $|A_{1,i_j} - A_1| \le \eta/4$.
We conclude that with the choices just made we have
\begin{equation*}
 \left|A_1 - \int_{\partial M_1} {\rm e}^{\mv_\infty}\, {\rm d}\mu_{h_{\kappa_1} }\right| \le \eta.
\end{equation*}
As $\eta$ is arbitrary, \eqref{4VI22.11rd} follows.

By Remark~\ref{NoSR22V22.1} and Deligne--Mumford compactness (cf., e.g., \cite[Proposition~A.2, Appendix~A.1]{RupflinToppingZhu}), the metric ${\rm e}^{\omega_\infty}h_{\kappa_1}$ is the punctured hyperbolic metric on $\partial M_1\setminus \{p_1\}$.
The Gauss--Bonnet theorem applies to such metrics and gives
\begin{equation*}
 A_1 = \red{2\pi}
 \big(
 1-\chi(\partial M_1)
 \big).
\end{equation*}
Passing to the limit $j\to\infty$ in the Gauss--Bonnet identity,
\begin{equation*}
 A_{1,i_j } + A_{2,i_j }
 =\int_{\zMone } {\rm e}^{\mv_{i_j} }\, {\rm d}\mu_{h_{\kappa_1}} +
 \int_{\zMtwo } {\rm e}^{\mv_{i_j}}\, {\rm d}\mu_{h_{\kappa_2}}
 = \red{2\pi}
 \big(
 2-\chi(\partial M_1) - \chi(\partial M_2)
 \big),
\end{equation*}
one obtains
\begin{equation*}
 A_2 = \red{2\pi}
 \big(
 1-\chi(\partial M_2)
 \big)
 {>0}
\end{equation*}
{since, by hypothesis, neither summand is a sphere. (Note that this argument fails for two-components gluing with one puncture at each summand and with one or two spherical summands.)}
So $A_2\ne 0$, and existence of a limiting conformal factor, realising a punctured metric on $\partial M_2\setminus \{p_2\}$, follows as before.

We have therefore established the equivalent of Lemma~5.9 of \cite{CDW} for both summands of the gluing construction.
The rest of the proof of Theorem~\ref{NoST29VII21.1}
is as in~\cite{CDW}.
\end{proof}

\section{Instabilities?} \label{s11VI22.1}

We start with the following observation:
Consider a pair of two-dimensional hyperbolic manifolds $\big(\twodimN _a,h_a\big)$, $a=1,2$ with a boundary (at finite distance) $\partial \twodimN _a$ with zero-mean curvature. Thus, the boundaries are closed curves which satisfy the geodesic equation. Suppose that the lengths of the boundary curves coincide.
Any two such manifolds can be identified at that boundary to yield a hyperbolic manifold in which the boundary curves become a closed geodesic.

The above allows us to provide a \emph{construction kit} for producing nontrivial manifolds with constant scalar curvature, higher-genus topology at infinity, and mass which is additive under Maskit gluing.

The simplest collection of the relevant building blocks is obtained as follows: Let $(M,g)$ be any ALH manifold with genus $\mathbf{g}_{\infty} \ge 1$. Let us denote by $\partial_\infty M$ the conformal boundary at infinity of $M$ and let $p\in \partial_\infty M$. We can carry out the construction of the proof of Theorem~\ref{NoST29VII21.1} by a gluing-at-infinity of a copy of $(M,g)$ to itself in a symmetric way, as in \cite{CDW},
and where the discrete parameter $1/i\to 0$ is replaced by a continuous parameter $\lambda \to 0$. There results a family of constant-scalar-curvature manifolds $(M_\lambda, g_\lambda)$ which are exactly hyperbolic in a neighborhood of a
totally geodesic two-dimensional half-sphere $\partial \mcU_{\lambda}$
 cutting $(M_\lambda,g_\lambda) $ in half
(in the example of Figure~\ref{F26VIII21.2} this is the half-sphere $\partial \mcU_{1,\epsilon} $ with $\epsilon=\lambda$). The resulting cut-in-half manifolds will be referred to as \emph{building blocks}.
The half-spheres $\partial \mcU_{\lambda}$ have a boundary at infinity which is a~closed geodesic $\gamma_\lambda$ cutting $\partial_\infty M_\lambda $
 in half, with the length $\ell_\lambda$ of $\gamma_\lambda$ varying continuously with~$\lambda$ by \cite[Remark~5.11]{CDW}, and with $\ell_\lambda$ tending to zero as $\lambda$ does.

Any building blocks with matching lengths of their closed geodesics at infinity can be joined together at the boundaries $\partial \mcU_{\lambda}$ to obtain a CSC ALH manifold with compact negatively curved conformal boundary. It should be clear from \eqref{18IX20.4} that the mass of the manifold so obtained
is the sum of the masses of the building blocks. Depending upon the sign of the mass of the second summand the new manifold can have a mass larger or smaller than that of the first summand.

The fact that a summand with positive mass can be replaced by one with negative one, thus lowering the mass of the connected manifold without changing the geometry of the other summand away from a small subset of its exactly hyperbolic region, suggests instability under time evolution:

\begin{Conjecture} \label{C12VI22.1}
 Three-dimensional CSC ALH manifolds with higher genus topology at conformal infinity and thin necks at conformal infinity are unstable.
\end{Conjecture}

\subsection*{Acknowledgements}

ED was supported by the grant ANR-17-CE40-0034 of the French National Research Agency ANR (project CCEM).

\pdfbookmark[1]{References}{ref}
\LastPageEnding

\end{document}